\documentclass{llncs}

\usepackage{verbatim}
\usepackage{citesort}
\usepackage{graphicx}
\usepackage{amsfonts}
\usepackage{times}
\usepackage{wrapfig}

\newtheorem{obs}[lemma]{Observation}

\renewenvironment{proof}{\medskip\noindent{\bf Proof:}}{\mbox{}\hfill\qed\par}

\newenvironment{sketch}{\medskip\noindent{\bf Proof Sketch:}}{\mbox{}\hfill\qed\par}

\title{On Contact Graphs with Cubes and Proportional Boxes}
\author
{
	Md.~Jawaherul Alam\inst{1} \and
	Michael Kaufmann\inst{2} \and
	Stephen G.~Kobourov\inst{1}
}

\institute{
	Department of Computer Science, University of Arizona, USA
	\and
	Wilhelm-Schickhard-Institut f\"ur Informatik, Universit\"at T\"ubingen, Germany
}

\begin{document}
\date{}
\maketitle

\begin{abstract}
We study two variants of the problem of contact representation of planar graphs with
 axis-aligned boxes. In a \textit{cube-contact representation} we realize each
 vertex with a cube, while in a \textit{proportional box-contact representation} each
 vertex is an axis-aligned box with a prespecified volume. We present
 algorithms for constructing cube-contact representation and proportional box-contact
 representation for several classes of planar graphs.

%A \textit{nested maximal outerplanar graph}
% is either a maximal outerplanar graph or a maximal planar graph in which the
% vertices on the outerface induce a maximal outerplanar graph and each component
% induced by internal vertices is another nested maximal outerplanar graph. We prove
% that these graphs admit proper contact representation with axis-aligned cubes in 3D.
\end{abstract}

\section{Introduction}

We study \textit{contact representations} of planar graphs in 3D, where vertices are
 represented by interior-disjoint axis-aligned boxes and edges are represented by shared
 boundaries between the corresponding boxes. A contact representation of a planar graph
 $G$ is \textit{proper} if for each edge $(u,v)$ of $G$, the boxes for $u$ and $v$ have a
 shared boundary with non-zero area. Such a contact between two boxes is also called a
 \textit{proper contact}. \textit{Cubes} are axis-aligned boxes where all sides have the same
 length. A contact representation of a planar graph with boxes is called a \textit{cube-contact}
 representation when all the boxes are cubes. In a weighted variant of
 the problem a \textit{proportional
 box-contact} representation is one where each vertex $v$ is represented with
 a box of volume $w(v)$, for any function
 $w:V\rightarrow\mathbb{R}^+$, assigning weights to the vertices $V$.
 Note that this ``value-by-volume'' representation is a natural
 generalization of the ``value-by-area'' cartograms in 2D.

%\subsection{Related Work}

\smallskip\noindent{\bf Related Work:} The history of representing planar graphs as contact graphs dates back at least to
 Koebe's 1930 theorem~\cite{Koebe36} for representing planar graphs by touching disks
 in 2D. Proper contact representation with rectangles in 2D is the well-known {\em rectangular
 dual} problem, for which several characterizations exist~\cite{KK85,Ungar53}.
 Representations with other axis-aligned and non-axis-aligned
 polygons~\cite{FM94,GHKK10,YS93} have been studied. Related graph-theoretic, combinatorial
 and geometric problems continue to be of
 interest~\cite{BGPV08,Fusy09,EMSV12}. The weighted variant
 of the problem has been considered in the context of 
 rectangular, rectilinear, and unrestricted cartograms~\cite{BR11,EFK+13,KS07}.%,Ueck-phd}.

Contact representations have been also considered in 3D. Thomassen~\cite{Thom88} shows that any planar graph has a proper
 contact representation with touching boxes, while Felsner and Francis~\cite{FF11} find
 a (not necessarily proper) contact representation of any planar graph with touching cubes.
 Recently, Bremner \textit{et al.}~\cite{BEF+12} asked whether any planar graph
 can be represented by proper contacts of cubes. They answered the question positively
 for the case of partial planar 3-trees and some planar grids, but the problem remains
 open for general planar graphs. The weighted variant of the problem
 in 3D is much less studied, although recently Alam \textit{et al.}~\cite{AKLPV14} have
 presented algorithms for proportional representation of several
 classes of graphs (e.g., outerplanar, planar bipartite, planar,
 complete), using 3D L-shapes.
%although there is a rich literature on the weighted
% contact representation in 2D~\cite{ourAlg13,ourDCG13,Ueck-phd}.

%\subsection{Our Contribution}

 \smallskip\noindent{\bf Our Contribution:} Here we expand the class of planar graph representable by proper contact of cubes. We also
 show that several classes of planar graphs admit proportional box-contact representations.
Specifically, we show how to compute a proportional box-contact representation for plane 3-trees, while a cube-contact
 representation for the same graph class follows
 from~\cite{BEF+12}. We also show how to compute a proportional
 box-contact representation and a cube-contact representation for {\em nested maximal outerplanar
   graphs}, which are defined as follows.
A \textit{nested outerplanar graph} is either an outerplanar graph or a planar graph $G$
 where each component induced by the internal vertices is another nested outerplanar
 graph with exactly three neighbors in the outerface of $G$. A \textit{nested maximal
 outerplanar graph} is a subclass of nested outerplanar graphs
% defined as follows. A \textit{nested maximal outerplanar graph} is 
that is either a maximal outerplanar graph or a
 maximal planar graph in which the vertices on the outerface induce a maximal outerplanar
 graph and each component induced by internal vertices is another nested maximal outerplanar
 graph. 
%We show that a nested maximal outerplanar graph admits proportional box-contact
 %representation for any specified weight-function, while any nested maximal outerplanar
 %graph has a proper contact representation with cubes.

%%%%%%%%%%%%%%%%%%%%\input{prelim.tex}

\section{Preliminaries}

%Here we give some definitions that we used throughout the paper.

A \textit{3-tree} is either a 3-cycle or a graph $G$ with a vertex $v$ of degree three in $G$
 such that $G-v$ is a 3-tree and the neighbors of $v$ form a triangle. If $G$ is planar, then
 it is called a \textit{planar 3-tree}. A \textit{plane 3-tree} is a planar 3-tree along with a fixed
 planar embedding. %It can be proved that 
Starting with a 3-cycle, any planar 3-tree can be
 formed by recursively inserting a vertex inside a face and adding an edge between the newly
 added vertex and each of the three vertices on the face~\cite{BE09,MNRA10}.
 Using this simple construction, we can create in linear time a {\em representative tree} for
 $G$~\cite{MNRA10}, which is an ordered rooted ternary tree $T_G$ spanning all the internal
 vertices of $G$. The root of $T_G$ is the first vertex we have to insert into the face of the three
 outer vertices. Adding a new vertex $v$ in $G$ will introduce three new faces belonging to
 $v$. The first vertex $w$ we add in each of these faces will be a child of $v$ in $T_G$. The
 correct order of $T_G$ can be obtained by adding new vertices according to the
 counterclockwise order of the introduced faces. 
%Note that for a planar 3-tree, a representative
% tree is an equivalent structure to the \textit{4-block tree} defined by Kant~\cite{Kant97}. 

An \textit{outerplanar graph} is one that has a planar embedding with all vertices
 on the same face (outerface). An outerplanar graph is \textit{maximal} if no edge can be added without violating its outerplanarity. Thus in a maximal outerplanar graph
 all the faces except for the outerface are triangles.
% A $k$-outerplanar graph is defined as follows. 
For $k>1$, a $k$-outerplanar
 graph $G$ is an embedded graph such that deleting the outer-vertices from $G$ yields a graph
 where each component is at most a $(k-1)$-outerplanar graph; a $1$-outerplanar graph is just an outerplanar graph.  Note that any planar graph is a
 $k$-outerplanar graph for some integer $k>0$. 
%The outerplanarity of a planar graph is  the minimum value of $k$ for which it is a $k$-outerplanar graph. 

%A \textit{nested-maximal outerplanar graph} is a maximal planar graph $G$
% where (i) the outer vertices of $G$ induces a maximal outerplanar graph, and
% (ii) deleting the outer vertices yield a graph whose components are nested
% maximal outerplanar graphs.
 Let $G$ be a planar graph. We define the \textit{pieces} of $G$ as follows. If $G$ is
 outerplanar, it has only one piece, the graph itself. Otherwise, let $G_1$, $G_2$, $\ldots$,
 $G_l$ be the components of the graph obtained by deleting the outer vertices (and their
 incident edges) from $G$. Then the pieces of $G$ are all the pieces of $G_i$ for each
 $i\in\{1, 2, \ldots, l\}$, as well as the subgraph of $G$ induced by the outer-vertices
 of $G$. Note that each piece of $G$ is an outerplanar graph. Since $G$ is an
 embedded graph, for each piece $P$ of $G$, we can define the \textit{interior} of $P$
 as the region bounded by the outer cycle of $P$. Then we can define a rooted tree
 $\mathcal{T}$ where the pieces of $G$ are the vertices of $\mathcal{T}$ and the
 parent-child relationship in $\mathcal{T}$ is determined as follows: for each piece
 $P$ of $G$, its children are all the pieces of $G$ that are in the interior of $P$ but
 not in the interior of any other pieces of $G$. A piece of $G$ has \textit{level} $l$
 if it is on the $l$-th level of $\mathcal{T}$. All the vertices of a piece
 at level $l$ are also \textit{$l$-level} vertices. A planar graph is a \textit{nested outerplanar
 graph} if each of its pieces at level $l>0$ has exactly three vertices of level $(l-1)$ as
 a neighbor of some of its vertices. On the other hand a \textit{nested maximal outerplanar
 graph} is a maximal planar graph where all the pieces are maximal outerplanar graphs.

\section{Representations for Planar 3-trees}

%Here we compute proportional box-contact representations of planar 3-trees.
Here we prove that planar 3-trees have proportional box-representations in two different ways.
The first one is a more intuitive proof; the second one includes a direct computation of the coordinates
for the representation.

\begin{theorem}
\label{th:p3t-box} Let $G=(V, E)$ be a plane 3-tree with a weight function $w$. Then a
 proportional box-contact representation of $G$ can be computed in linear time.
%\label{th:p3t-box} Let $G=(V, E)$ be a plane 3-tree and let $w:V\rightarrow \mathbb{R}^+$
% be a weight function. Then a proportional box-contact representation of $G$ can be
% computed in linear time.
\end{theorem}

\begin{figure}[t]
%\vspace{-0.8cm}
\centering
\includegraphics[width=0.8\textwidth]{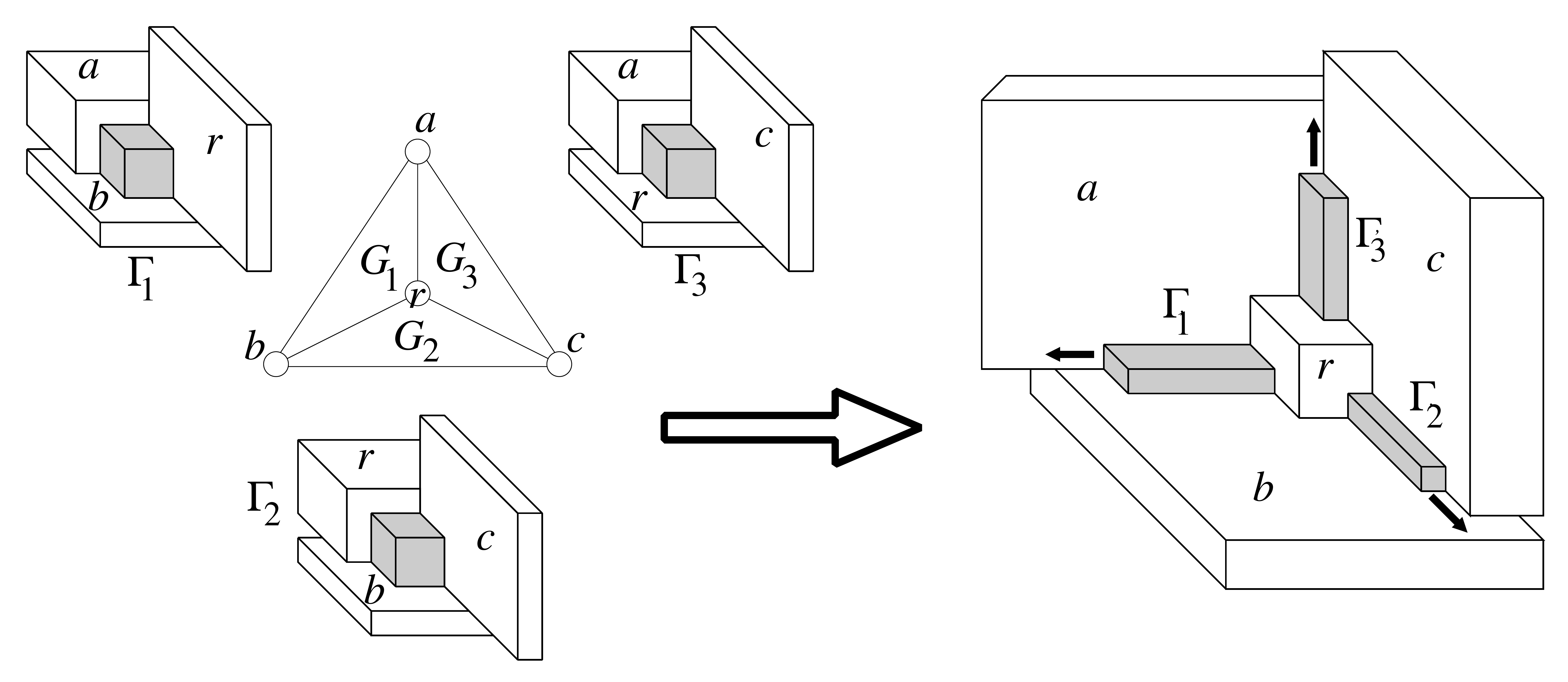}
\caption{Illustration for the proof of Theorem~\ref{th:p3t-box}.}
\label{fig:prop-p3t-2}
\end{figure}

%\begin{proof}
\noindent
\textbf{First Proof:}
Let $a$, $b$, $c$ be the outer vertices of $G$. We construct a representation $\Gamma$ for
 $G$ where $b$ occupies the bottom side of $\Gamma$, $a$ occupies the back of
 $\Gamma-\{b\}$ and $c$ occupies the right side of $\Gamma-\{a,b\}$; see
 Fig.~\ref{fig:prop-p3t-2}. 
 Here for a set of vertices $S$, $\Gamma-S$ denotes the representation obtained from $\Gamma$ by deleting
 the boxes representing the vertices in $S$. The claim is trivial when
 $G$ is a triangle, so assume that $G$ has at least one internal vertex. Let $r$
 be the root of the representation tree $T_G$ of $G$. Then $r$ is adjacent to $a$, $b$ and $c$
 and thus defines three regions $G_1$, $G_2$ and $G_3$ inside the triangles $\Delta_1=abr$,
 $\Delta_2=bcr$ and $\Delta_3=car$, respectively (including the vertices of these triangles).
 By induction hypothesis $G_i$, $i=1,2,3$ has a proportional box-contact representation
 $\Gamma_i$ where the boxes for the three vertices in $\Delta_i$ occupy the bottom, back
 and right sides of $\Gamma_i$. Define $\Gamma'_i=\Gamma_i-\Delta_i$. We now construct
 the desired representation for $G$. First take a box for $r$ with volume $w(r)$ and place it in
 a corner created by the intersection of three pairwise-touching boxes; see
 Fig.~\ref{fig:prop-p3t-2}.  For each $\Delta_i$, $i=1,2,3$, there is a corner $p_i$ formed by the
 intersection of the three boxes for $\Delta_i$. We now place $\Gamma'_i$ (after possible
 scaling) in the corner $p_i$ so that it touches the boxes for the vertices in $\Delta_i$ by
 three planes. Note that this is always possible since we can choose the surface areas for $a$, $b$ and $c$ to
 be arbitrarily large and still realize their corresponding weights
% for these vertices
 by appropriately changing the thickness in the third dimension.
This construction requires only linear time, by keeping the scaling factor for each
 region in the representative tree $T_G$ at the vertex representing
 that region. 
% by a  linear-time
% for the representation
Then the exact coordinates can be computed with a top-down traversal of $T_G$. \qed
%\end{proof}

\noindent
\textbf{Second Proof:}
 Assume (after possible factoring) that for each vertex $v$ of $G$, the weight $w(v)$ is at least
 1. 
Let $T_G$ be the representative tree of $G$. For
 any vertex $v$ of $T_G$, we denote by $U_v$, the set of the descendants of $v$ in $T_G$
 including $v$. The \textit{predecessors} of $v$ are the neighbors of $v$ in $G$ that are not
 in $U_v$. Clearly each vertex of $T_G$ has exactly three predecessors.
%
%
%Let $T_G$ be the representative tree of $G$ and for each internal vertex $v$ of $G$, let
% $U_v$ denote the set of the descendants of $v$ in $T_G$. 
We now define a parameter $W(v)$
 for each vertex $v$ of $T_G$. Let $v_1$, $v_2$ and $v_3$ be the three children of $v$ in
 $T_G$ (where zero or more of these three children may be empty). Then $W(v)$ is defined as
 $\Pi_{i=1}^{3}[W(v_i)+\sqrt[3]{w(v)}]$, where $U(v_i)$ is taken as zero when $v_i$ is empty. We
 can compute the value of $W(v)$ for each vertex $v$ of $T_G$ by a linear-time bottom-up
 traversal of $T_G$. Once we have computed these values, we proceed on constructing the
 box-contact representation as follows.

\begin{figure}[htbp]
\centering
\includegraphics[width=.9\textwidth]{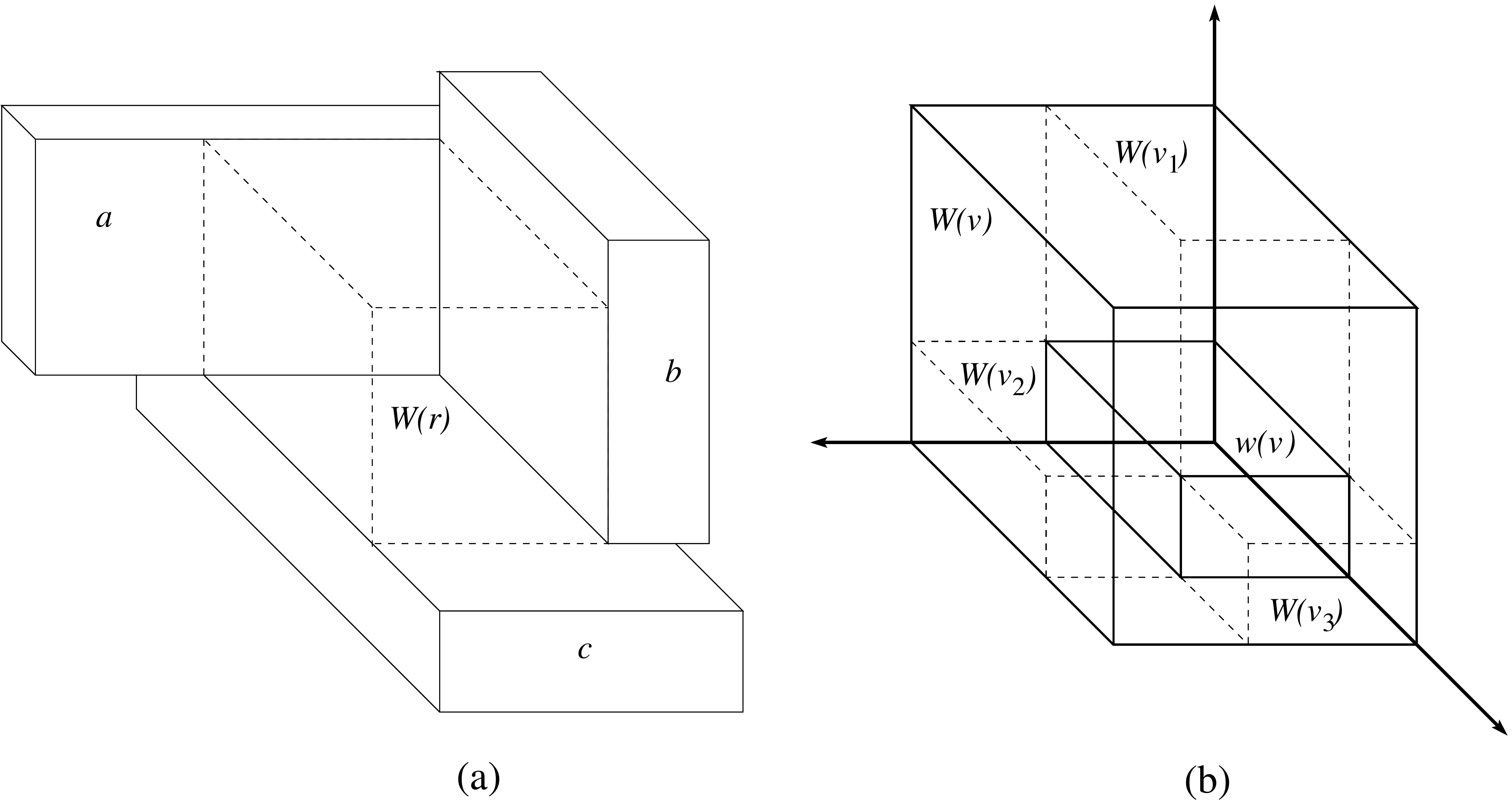}
\caption{Illustration for the second proof of Theorem~\ref{th:p3t-box}.}
\label{fig:prop-p3t}
\end{figure}

Let $a$, $b$, $c$ be the three outer vertices of $G$ in the clockwise order and let $r$ be the
 root of $T_G$. We start by computing three boxes for $a$, $b$ and $c$ with the correct volume
 as illustrated in Fig.~\ref{fig:prop-p3t}(a), so that the volume of the dotted box $R$ is $W(r)$.
 We will now construct a box representation of $U_r$ inside $R$ so that all the vertices in $U_r$
 adjacent to an outer vertex is represented by a box with a face co-planar on the face of $R$
 adjacent to box representing that outer vertex. We do this recursively by a top-down
 computation on $T_G$. Let $v$ be a vertex of $T_G$ with the three predecessors $u_1$, $u_2$
 and $u_3$. Let $D(v)$ be a box with volume $W(v)$ and let $t_1$, $t_2$, $t_3$ be three faces
 of it with a common point. While traversing $v$, we compute a proportional box-contact
 representation of $U_v$ inside $D(v)$ where the vertices in $U_v$ adjacent to $u_i$ for some
 $i\in\{1, 2, 3\}$ is represented by a box with a face co-planar with $t_i$. Let $v_1$, $v_2$ and
 $v_3$ are the three children of $v$ in $T_G$ (where zero or more of these children may be
 empty). Also assume that $x_1$, $x_2$, $x_3$ are the length, width and height of $D(v)$,
 respectively and $p$ is the common point of $t_1$, $t_2$ and $t_3$. Then first compute a
 box $R(v)$ of volume $w(v)$ for $v$ with a corner at $p$ where $x'_1$, $x'_2$ and $x'_3$
 are the length, width and height of $R(v)$, such that
 $x'_i= \frac{\sqrt[3]{w(v)}}{W(v_1)+\sqrt[3]{w(v)}}$.
 These choices of $x'_i$'s also creates three boxes $D(v_i)$ with volume at least $W(v_i)$,
 $i\in\{1, 2, 3\}$, as illustrated in Fig.~\ref{fig:prop-p3t}(b). Finally we recursively compute
 the box representations for $U_{v_i}$ inside $D(v_i)$ for $i\in\{1, 2, 3\}$ to complete
 the construction. \qed

\begin{theorem}~\cite{BEF+12}
Let $G$ be a plane 3-tree. Then a cube-contact representation of $G$ can be computed in linear
 time.
\end{theorem}

The proof of this claim also relies on the recursive decomposition of
planar 3-trees. 
%Using similar idea to the proof for Theorem~\ref{th:p3t-box}, Bremner \textit{et
 %al.}~\cite{BEF+12} gave an algorithm for cube-contact representation for a plane 3-tree.

%%%%%%%%%%\input{nested.tex}

\section{Cube-Contacts for Nested Maximal Outerplanar Graphs}
%\section{Cube-Contacts of Nested Maximal Outerplanar Graphs}
%\section{Cube-Representations for Nested Maximal Outerplanar Graphs}

We prove the following main theorem in this section:

\begin{theorem}
Any nested maximal outerplanar graph has a proper contact representation with cubes.
\label{th:nested}
\end{theorem}

We prove Theorem~\ref{th:nested} by construction, starting with a representation
 for each piece of $G$, and combining the pieces to complete the
 representation for $G$.

Let $G$ be a nested maximal outerplanar graph. We first augment the graph $G$ by adding
 three mutually adjacent dummy vertices $\{A, B, C\}$ on the outerface
 and then triangulating the graph by adding dummy edges from these three
 vertices to the outer vertices of $G$ such that the graph remains planar; see
 Fig.~\ref{fig:cube-merge1}(a). Call this graph the \textit{extended graph} of $G$.
 For consistency, let the three dummy vertices have level $0$.
 The observation below follows from the definition of nested maximal outerplanar graphs.

\begin{obs}
	\label{obs:extended} Let $G$ be a nested-maximal planar graph and let $G'$ be the
	extended graph of $G$. Then for each piece $P$ of $G$ at level $l$, there is a triangle
	of $(l-1)$-level vertices adjacent to the vertices of $P$ and no other $k$-level
	vertices with $k<l$ are adjacent to any vertex of $P$.
\end{obs}

Given this observation, we use the following strategy to obtain a contact representation of $G$
 with cubes. For each piece $P$ of $G$ at level $l$, let $A$, $B$ and $C$ be the three
 $(l-1)$-level vertices adjacent to $P$'s vertices. Let $P'$ be
 the subgraph of $G$
 induced vertices of $P$ as well as $A$, $B$ and
 $C$; call $P'$ the {\em extended
 piece} of $G$ for $P$. We obtain a contact representation of $P'$ with cubes and delete
 the three cubes for $A$, $B$ and $C$ to obtain the contact representation of $P$ with cubes.
 Finally, we combine the representations for the pieces to complete the desired
 representation of $G$.

Before we give more details on this algorithm, we have the following lemma, that we use in this section.
 Furthermore this result is also interesting by itself, since for any outerplanar graph $O$, where each
 face has at least one outer edge, Lemma~\ref{lem:square} gives a contact representation of $O$ on
 the plane with squares such that the outer boundary of the representation is a rectangle.

\begin{lemma}
	\label{lem:square} Let $G$ be planar graph with outerface $ABba$ and at least one
	internal vertex, such that $G-\{A,B\}$ is a maximal outerplanar graph. If there is no
	chord between any two neighbors of $A$ and no chord between any two neighbors of
	$B$, then $G$ has a contact representation $\Gamma$ in 2D where each inner vertex
	is represented by a square, the union of these squares forms a rectangle,
	and the four sides of these rectangles represent $A$, $B$, $b$ and $a$, respectively.
\end{lemma}
\begin{proof}
 We prove this lemma by induction on the number of vertices in $G$. Denote the
 maximal outerplanar graph $H=G-\{A,B\}$; see Fig.~\ref{fig:square}(a). If $G$ contains
 only one internal vertex $v$, then we compute $\Gamma$ by
 representing $v$ by a square $R(v)$ of arbitrary size and representing $A$, $B$, $b$
 and $a$ by the left, bottom, right and top sides of $R(v)$.

\begin{figure}[t]
%\vspace{-0.8cm}
	\centering
	\includegraphics[width=0.8\textwidth]{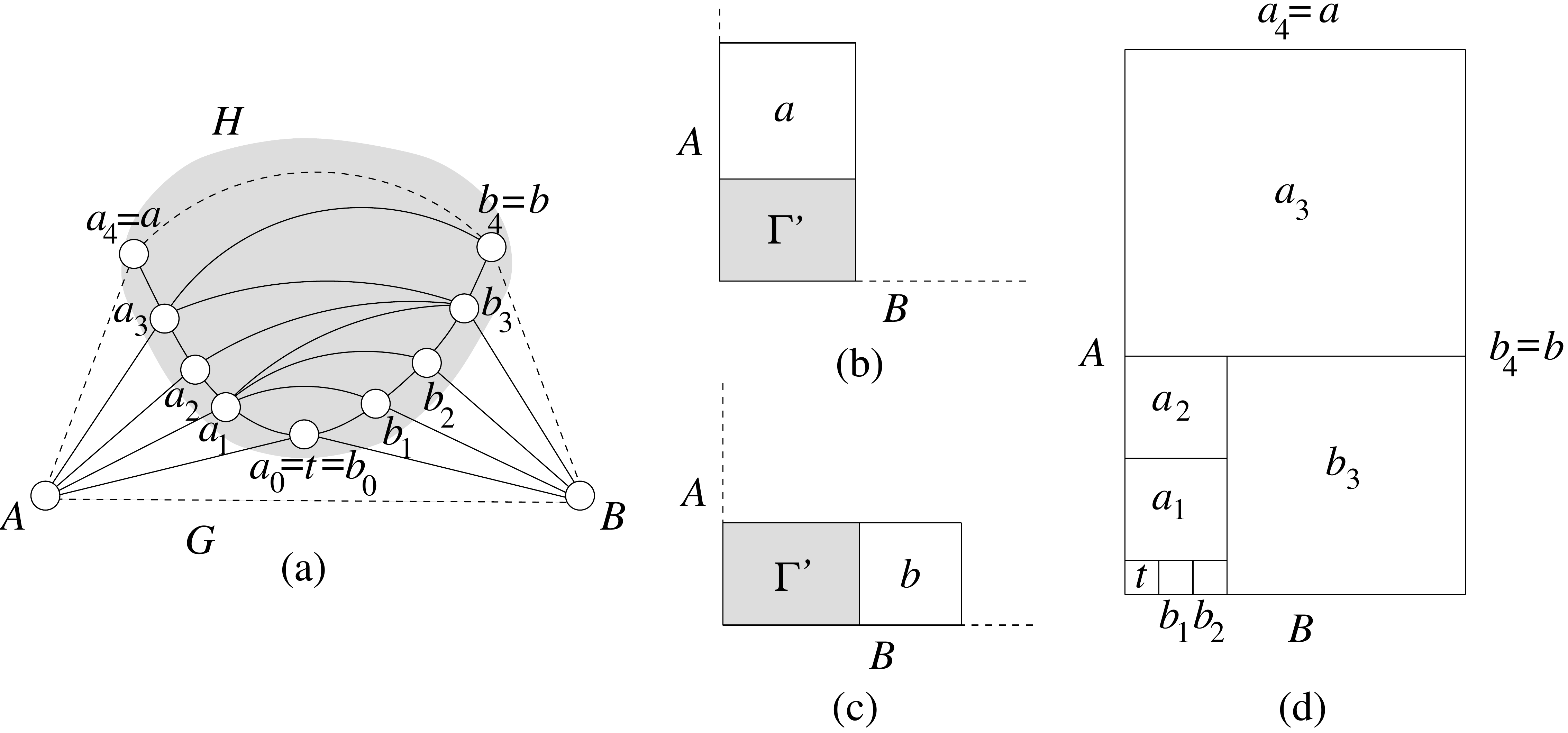}
	\caption{Illustration for the proof of Lemma~\ref{lem:square}.}
	\label{fig:square}
\end{figure}

 We thus assume that $G$ has at least two internal vertices. Let $u$ be the unique common
 neighbor of $\{a, b\}$ in $H$. If $u$ is a neighbor of $A$, then $H-\{a\}$ is a maximal
 outerplanar graph. By induction hypothesis, $G-\{a\}$ has a contact representation
 $\Gamma'$ where each internal vertex of $G-\{a\}$ is represented by a square and the left,
 bottom, right and top sides of $\Gamma'$ represent $A$, $B$, $b$ and $u$. Then we
 compute $\Gamma$ from $\Gamma'$ by adding a square $R(u)$ to represent $u$ such
 that $R(u)$ spans the entire width of $\Gamma'$ and is placed on top of $\Gamma'$; see
 Fig.~\ref{fig:square}(b). A similar construction can be used if $u$ is a neighbor of $B$; see
 % Similarly if $u$ is a neighbor of $B$, then by induction hypothesis
% we compute a contact representation $\Gamma'$ of $G-\{b\}$ where each
% internal vertex of $G-\{b\}$ is represented by a square and the left, bottom, right and top
% sides of $\Gamma'$ represent $A$, $B$, $u$ and $a$. Then we compute
% $\Gamma$ from $\Gamma'$ by adding a square $R(u)$ to represent $u$ such that $R(u)$
% spans the entire height of $\Gamma'$ and is placed to the right of $\Gamma'$; see
 Fig.~\ref{fig:square}(c).
% Since $u$ must be neighbor of either $A$ or $B$,
 We thus compute a contact representation for $G$; see Fig.~\ref{fig:square}(d).
\end{proof}

\subsection{Cube-Contact Representation for Extended Pieces}
%\subsection{Representations for Extended Pieces}

%Here we give an algorithm to compute the contact representation of an extended piece of $G$.
% We have the following lemma.

\begin{lemma}
	\label{lem:piece} Let $P$ be a piece of $G$ at level $l$ and $P'$ be the
	extended piece for $P$ with $(l-1)$-level vertices $A$, $B$, $C$. Then $P'$ has a
	cube-contact representation.
\end{lemma}
\begin{proof} Let $r$ be a common neighbor of $B$ and $C$; $s$ a common neighbor of $A$
 and $C$; $t$ a common neighbor of $A$ and $B$. It is easy to find a contact
 representation of $P'$ if $r$, $s$ and $t$ are the only vertices of
 $P$, so let
 $P$ have at least four vertices. The outer cycle of $P$ can be
 partitioned into three paths: $P_a$ is the path from $s$ to $t$, $P_b$ is the path from $r$ to
 $t$ and $P_c$ is the path from $r$ to $s$. Note that all vertices on the path $P_a$ ($P_b$,
 $P_c$) are adjacent to $A$ ($B$, $C$). A chord $(u,v)$ is a \textit{short chord} if it is between
 two vertices on the same path from the set $\{P_a, P_b, P_c\}$. (Note that a chord between two
 vertices from the set $\{r,s,t\}$ is also a short chord.) We have the following two cases.

\smallskip\noindent
\textbf{Case A: There is no short chord in $P$.} In this case all the chords of $P$ are between two
 different paths. We consider the following two subcases.

% We first assume that there is
% no chord between any two vertices on the same path from the set $\{P_a, P_b, P_c\}$. Thus all
% the chords are between two different paths. We consider the following two cases.

\begin{figure}[tb]
%\vspace{-0.8cm}
	\centering
	\includegraphics[width=0.7\textwidth]{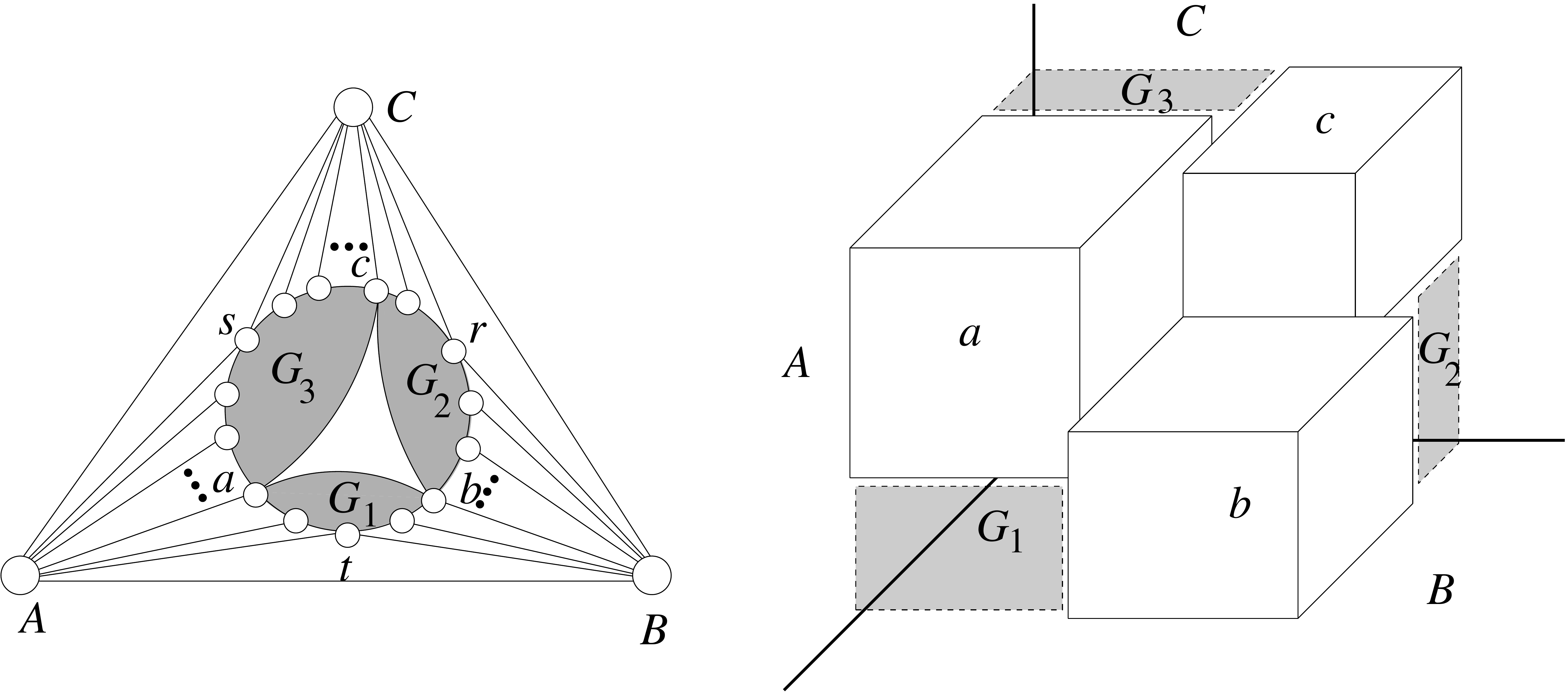}\\
	(a)\hspace{0.38\textwidth}(b)\hspace{0.05\textwidth}
	\caption{Illustration for \textbf{Case A1} in the proof of Lemma~\ref{lem:piece}.}
	\label{fig:cube-merge1}
\end{figure}

%\smallskip\noindent
\textbf{Case A1: There is no chord with one end-point in $\{r, s, t\}$.} In this case, due to
 maximal-planarity there exist three vertices $a$, $b$ and $c$, adjacent to $A$, $B$, and $C$,
 respectively such that (i) $ab$ is the chord between vertices of $P_a$ and $P_b$ farthest
 away from $t$, (ii) $bc$ is the chord between vertices of $P_b$ and $P_c$ farthest away
 from $r$, and (iii) $ac$ is the chord between vertices of $P_a$ and $P_c$ farthest away
 from $s$; see Fig.~\ref{fig:cube-merge1}(a). We can then find three interior-disjoint
 subgraphs of $P'$ defined by three
% separating
 cycles of $P'$: $G_1$ is the one induced by
 all vertices on or inside $ABba$; $G_2$ is induced by all vertices on or inside $BCcb$;
 and $G_3$ is induced by all vertices on or inside $ACca$. Each of these subgraphs
 has the common property that if we delete two vertices from the outerface (two vertices from
 the set $\{A, B, C\}$ in each subgraph), we get an outerplanar graph.
From the representation with squares from the proof of Lemma~\ref{lem:square}, we find a contact
 representation of $G_i$, $i=1,2,3$ where each internal vertex of $G_i$ is represented by a cube
 and the union of all these cubes forms a rectangular box whose four sides realize the outer vertices.
 We use such a representation to obtain a contact representation of $P'$ with cubes as follows.

We draw pairwise adjacent cubes (of arbitrary size) for $A$, $B$,
$C$. We need to place the cubes for all the vertices of $P$ in the a
corner defined by three faces of the cubes for $A$, $B$, $C$. 
%Therefore we may assume that $A$, $B$ $C$ are represented by three mutually adjacent
 %walls at right angles from each other. 
Then we place three mutually touching cubes for
 $a$, $b$ and $c$, which touch the walls for $A$, $B$ and $C$, respectively; see
 Fig.~\ref{fig:cube-merge1}(b). We also compute a contact representation of the internal
 vertices for each of the
 three graphs $G_1$, $G_2$ and $G_3$ with cubes using Lemma~\ref{lem:square}, so that
 the outer boundary for each of these representation forms a rectangular pipe. We adjust
 the sizes of the three cubes for $a$, $b$ and $c$ in such a way that the three highlighted
 rectangular pipes precisely fit these three representations (after some possible scaling).
 Note that this construction works even if one or more of the subgraphs $G_1$, $G_2$ and
 $G_3$ are empty. This completes the analysis of \textbf{Case A1}.

\begin{figure}[tb]
%\vspace{-0.8cm}
	\centering
	\includegraphics[width=0.7\textwidth]{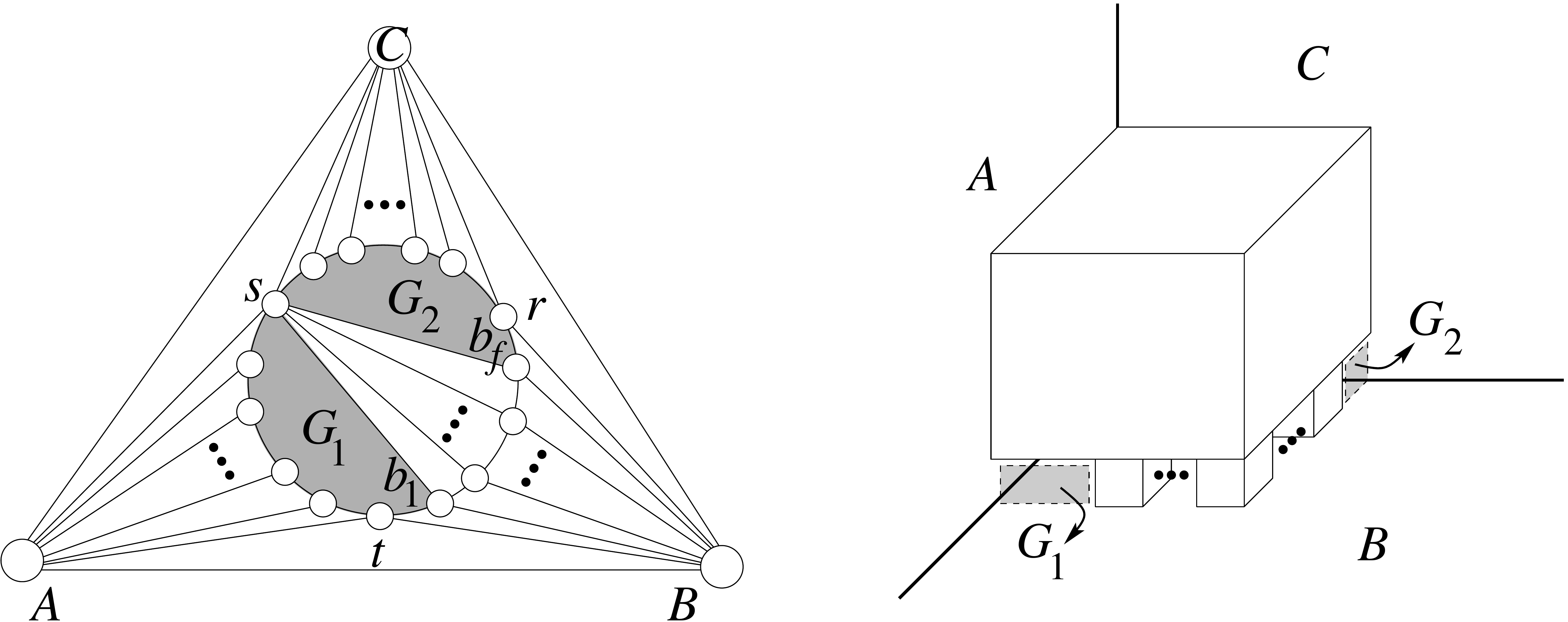}\\
	(a)\hspace{0.38\textwidth}(b)\hspace{0.05\textwidth}
	\caption{Illustration for \textbf{Case A2} in the proof of Lemma~\ref{lem:piece}.}
	\label{fig:cube-merge2}
\end{figure}

%\smallskip\noindent
\textbf{Case A2: There is at least one chord with one end-point in $\{r, s, t\}$.} Due to planarity
 all such chords will have the same end-point in $\{ r,s,t\}$. Suppose $s$ is this common end
 point for these chords; see Fig.~\ref{fig:cube-merge2}(a).
 Let $b_1$ and $b_f$ be the first and last endpoints in the clockwise order of these chords
 around $s$. Then we can find two subgraphs $G_1$ and $G_2$
 induced by the vertices on or inside two separating cycles $ABb_1s$ and $BCsb_f$. We
 find contact representations for the internal vertices of these two graphs $G_1$ and $G_2$
 using Lemma~\ref{lem:square} so that the outer-boundaries of these representation form
 rectangular pipes. We then obtain the desired contact representation for $P'$, starting with the
 three mutually touching walls for $A$, $B$ and $C$ at right angles
 from each other, placing
 the cubes for $s$ and $b_1$, $\ldots$, $b_f$ as illustrated in
 Fig.~\ref{fig:cube-merge2}(b), and fitting the representations for $G_1$ and $G_2$ (after some
 possible scaling) in the highlighted regions.

%This completes the construction for the contact representation of $P'$ with the assumption
% that there is no chord between any two vertices on the same path from $P_a$, $P_b$, $P_c$.
% Finally now consider the case where there might be chords between vertices on the same path.

\smallskip\noindent
\textbf{Case B: there are some shord chords in $P$.}
 In this case, we find at most four subgraphs from $P'$ as follows. At each path in
 $\{ P_a, P_b, P_c\}$, we find the \textit{outermost chord},
 i.e., one that is not contained inside any other chords on the same path.
 Suppose these chords are $a_1a_2$,
 $b_1b_2$ and $c_1c_2$, on the three paths $P_a$, $P_b$, $P_c$, respectively.
 Then three of these subgraphs $G_a$, $G_b$ and $G_c$ are  induced by the vertices
 on or inside the three
% separating
 triangles $Aa_1a_2$, $Bb_1b_2$ and $Cc_1c_2$.
 The fourth subgraph $P^*$ is obtained from $P'$ by deleting all the inner
 vertices of the three graphs $G_a$, $G_b$ and $G_c$; see Fig.~\ref{fig:short-chord}.

\begin{figure}[t]
%\vspace{-1cm}
	\centering
	\includegraphics[width=0.9\textwidth]{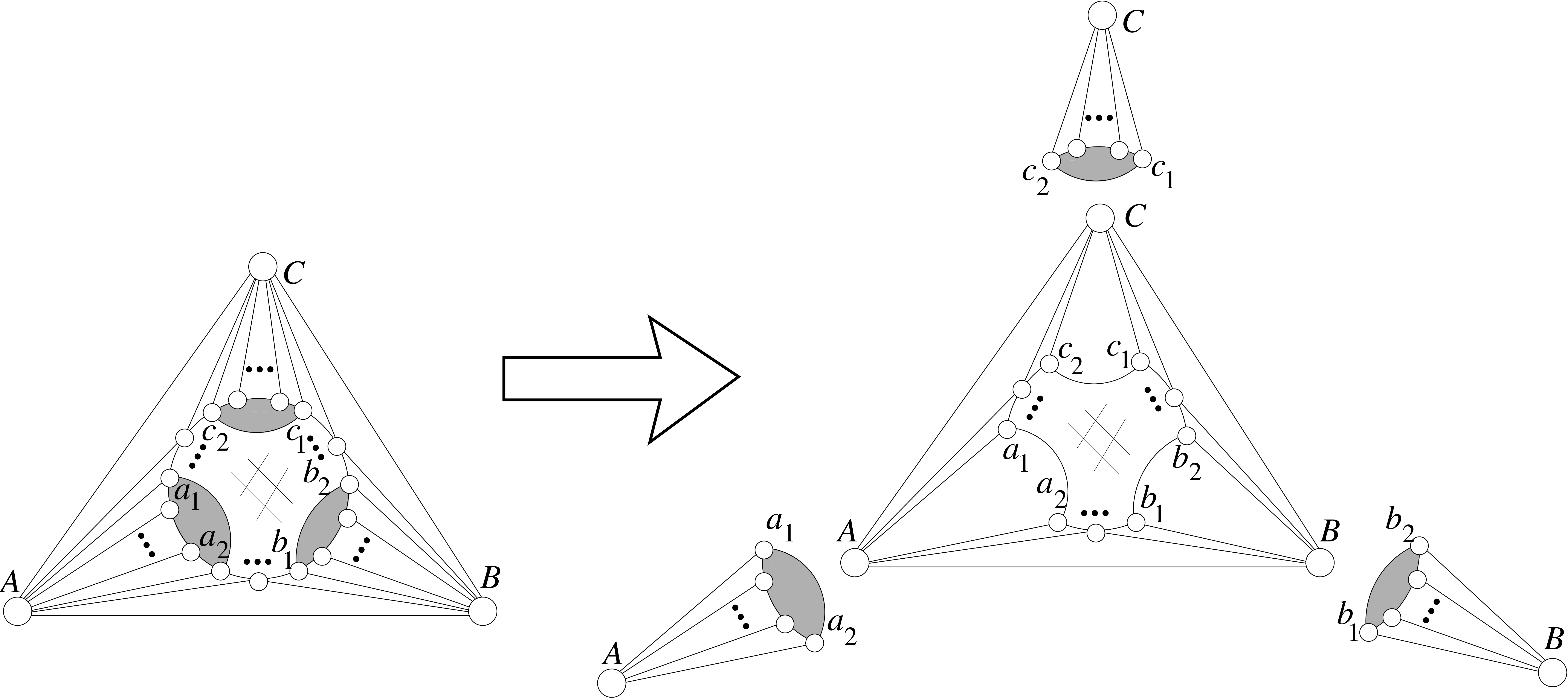}
	\caption{Removing chords with end-vertices in the same neighborhood.}
	\label{fig:short-chord}
\end{figure}

 A cube representation of $P^*$ can be found by the algorithm in \textbf{Case A}, as $P^*$
 fits the condition that there is no chord between any two neighbors of the same vertex in
 $\{A, B, C\}$. Note that by moving the cubes in the representation by an arbitrarily
 small amount, we can make sure that for each triangle $xyz$ in $P^*$, the three cubes for $x$,
 $y$ and $z$ form a corner surrounded by three mutually touching walls at right angles to
 each other. Now observe that each of the three graphs $G_a$, $G_b$ and $G_c$ is a planar
 3-tree; thus using the algorithm of either~\cite{BEF+12} or~\cite{FF11}, we can
 place the internal vertices of these three graphs in their corresponding corners, thereby
 completing the representation.
\end{proof}

\subsection{Cube-Contact Representation for a Nested Maximal Outerplanar Graph}
%\subsection{Representations for Nested Maximal Outerplanar Graphs}

%\begin{comment}
%We are now ready to prove Theorem~\ref{th:nested}.
%\end{comment}
%\bigskip
\noindent
\textbf{Proof of Theorem~\ref{th:nested}:} Let $G$ be a nested maximal outerplanar graph.
 We build the contact representation of $G$ by a top-down traversal of the rooted tree
 $\mathcal{T}$ of the pieces of $G$.
 We start by creating a corner surrounded by three mutually touching walls at right angle
 to each other. Then whenever we traverse any vertex of $\mathcal{T}$, we realize the
 corresponding piece $P$ at level $l$ by obtaining a representation using
 Lemma~\ref{lem:piece} and placing  this in the corner created by the three already-placed
 cubes for the three $(l-1)$-level vertices adjacent to $P$ (after possible scaling). \qed
% Suppose we want to draw a piece
% $P$ at level $l$. Let $P'$ be the corresponding extended piece. We may assume that the
% three $(l-1)$-level vertices have already been represented by three cubes creating three
% mutually touching walls at right angle to each other. Then we find a representation
% $\Gamma(P')$ for $P'$ using Lemma~\ref{lem:piece}. We then delete the three cubes for
% the three $(l-1)$-level vertices from $\Gamma(P')$ to find a representation for $P$
% and we place this representation in the corner created by the three cubes for the
% $(l-1)$-level vertices (after possible scaling). When the entire tree has been traversed
% we have a proper contact representation of $G$ with cubes. \qed

%%%%%%%%%%%%%\input{nested-prop.tex}

\section{Proportional Box-Contacts for Nested Outerplanar Graphs}
%\section{Proportional Contacts for Nested Outerplanar Graphs}
%\section{Proportional Representations for Nested Outerplanar Graphs}

In this section we prove the following main theorem.

\begin{theorem}
\label{th:nested-prop}
	Let $G=(V,E)$ be a nested outerplanar graph and let $w:V\rightarrow\mathbb{R}^+$
	 be a weight function defining weights for the vertices of $G$. Then $G$ has a
	 proportional contact representation with axis-aligned boxes with respect to $w$.
\end{theorem}

We construct a proportional representation for $G$ using a similar strategy as in the previous
 section: we traverse the construction tree $\mathcal{T}$ of $G$ and deal with each piece
 of $G$ separately. Each piece $P$ of $G$ is an outerplanar graph and hence one can easily
 construct a proportional box-contact representation for $P$ as follows. Any outerplanar
 graph $P$ has a contact representation with rectangles in the plane. In fact in~\cite{ourAlg13},
 it was shown that $P$ has a contact representation with rectangles
 on the plane where the rectangles realize prespecified weights by
 their areas. Thus by giving unit heights to all rectangles we can
 obtain a proportional box-contact representation of $P$ for any given
 weight function. 
%we can
% find such a contact representation with rectangles in the plane for $P$ that realizes the
% weight function $w$ restricted to the vertices in $P$ and then we 
However if we construct proportional box-contact representation for each piece of $G$ in this way,
it is not clear that we can combine them all to find a proportional contact representation of the
 whole graph $G$. 
%We still use the idea for the construction of a proportional contact
% representation of an outerplanar graph with rectangles in the plane from~\cite{ourAlg13}.
Instead, we use this construction idea in Lemmas~\ref{lem:stair}~and~\ref{lem:double-stair} to
build two different proportional rectangle-contact representations for
outerplanar graphs
% with rectangles in  the plane
% We call these two layouts ``staircase layout'' and ``double-staircase layout'', respectively
 and we use them in the proof of Theorem~\ref{th:nested-prop}. 

Suppose $O$ is an outerplanar graph and $\Gamma$ is a contact representation of $O$ with rectangles
 in the plane. We say that a corner of a rectangle in $\Gamma$ is \textit{exposed} if it is
 on the outer-boundary of $\Gamma$ and is not shared with any other rectangles.
% We now have the following two lemmas.

\begin{lemma}
\label{lem:stair}
	Let $O$ be a maximal outerplanar graph with a weight function $w$. Let
	$1$, $\ldots$, $n$ be the clockwise order of the vertices around the
	outer-cycle. Then a proportional rectangle-contact representation $\Gamma$ of
	$O$ for $w$ can be computed so that rectangle $R_1$
	for $1$ is leftmost in $\Gamma$, rectangle $R_n$ for $n$ is
        bottommost in $\Gamma - R_1$, and the top-right corner for each rectangle is exposed in $\Gamma$.
\end{lemma}
\begin{comment}
	\begin{figure}[t]
	%\vspace{-0.8cm}
		\centering
		\includegraphics[width=0.4\textwidth]{stair.pdf}
		\caption{Illustration for the proof of Lemma~\ref{lem:stair}.}
	\label{fig:stair}
	\end{figure}
\begin{sketch}
Constructing $\Gamma$ is easy when $G$ is a single edge $(1,n)$, so
let $G$ contain at least 3 vertices. Let $x$ be the unique vertex adjacent to $(1,n)$. Denote by
 $G[1,x]$ the graph induced by all vertices between $1$ and $x$ and by $G[x,n]$ the graph
 induced by the vertices between $x$ and $n$. Recursively draw $G[1,x]$ and $G[x,n]$ and
 remove the rectangles for $1$, $x$, $n$ to find the drawings $\Gamma_1$ and $\Gamma_2$,
 respectively. Finally draw the rectangles $R_1$, $R_x$ and $R_n$ for $1$, $x$ and $n$,
 with the required areas and place $\Gamma_1$ (after possible scaling) between
 $R_1$, $R_x$ and $\Gamma_2$ (after possible scaling)  between $R_x$, $R_n$ to complete
 the drawing; see Fig.~\ref{fig:stair}.
\end{sketch}
\end{comment}

\begin{proof}
We give an algorithm that recursively computes $\Gamma$. Constructing $\Gamma$ is easy
 when $G$ is a single edge $(1,n)$. We thus assume that $G$ has at least 3 vertices. Let $x$
 be the (unique) third vertex on the inner face that is adjacent to $(1,n)$. Then graph $G$ can
 be split into two graphs at vertex $x$ and edge $(1,n)$: $G[1,x]$ consists of the graph induced
 by all vertices between $1$ and $x$ in clockwise order around the outer-cycle; while $G[x,n]$
 consists of the graph induced by the vertices between $x$ and $n$.

%	\begin{figure}[htbp]
%		\centering
%		\includegraphics[width=0.75\textwidth]{stair.eps}
%		\caption{Illustration for the proof of Lemma~\ref{lem:stair}.}
%	\label{fig:stair}
%	\end{figure}

Recursively draw $G[1,x]$ and remove the rectangles for $1$ and $x$ from it; call the result
 $\Gamma_1$. Again recursively draw $G[x,n]$ and remove $x$ and $n$ from it; call the result
 $\Gamma_2$. Now draw a rectangle $R_x$ for $x$ with area $w(x)$. Let $l_x$ and $h_x$ be
 the width and height of $R_x$, respectively. Then draw the rectangles $R_1$ and $R_n$ for
 $1$ and $n$ touching the left and the bottom sides of $R_x$, respectively with necessary areas.
 Select the widths and heights of these two rectangles such that the area $l_x(h_1-h_n-h_x)$
 can contain $\Gamma_1$ while the area $(w_n-w_x)*h_x$ can contain $\Gamma_2$, where
 $l_j$ and $h_j$ denote the width and height of $R_j$, respectively for $j\in\{1,n\}$. Finally
 place $\Gamma_1$ (after possible scaling) touching the right side of $R_1$ and the top side
 of $R_x$ and place $\Gamma_2$ (after possible scaling) touching the right side of $R_x$ and
 the top side of $R_n$ to complete the drawing; see Fig.~\ref{fig:stair}.
\end{proof}

\begin{figure}[tb]
\hfill
\begin{minipage}[b]{.4\textwidth}
	\centering
	\includegraphics[width=\textwidth]{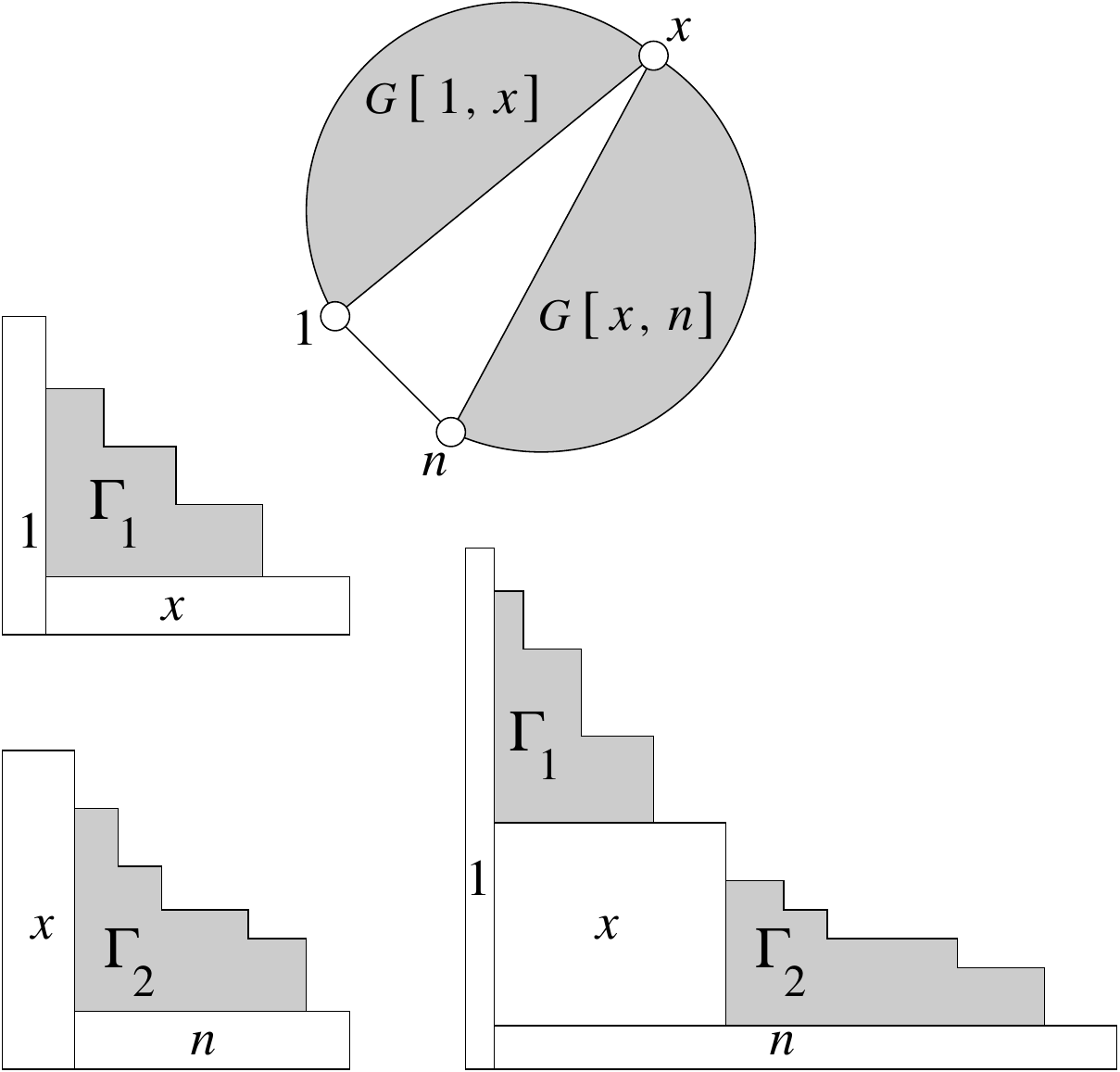}
	\caption{Illustration for the proof of Lemma~\ref{lem:stair}.}
	\label{fig:stair}
\end{minipage}
\hfill
\begin{minipage}[b]{.4\textwidth}
	\centering
	\includegraphics[width=\textwidth]{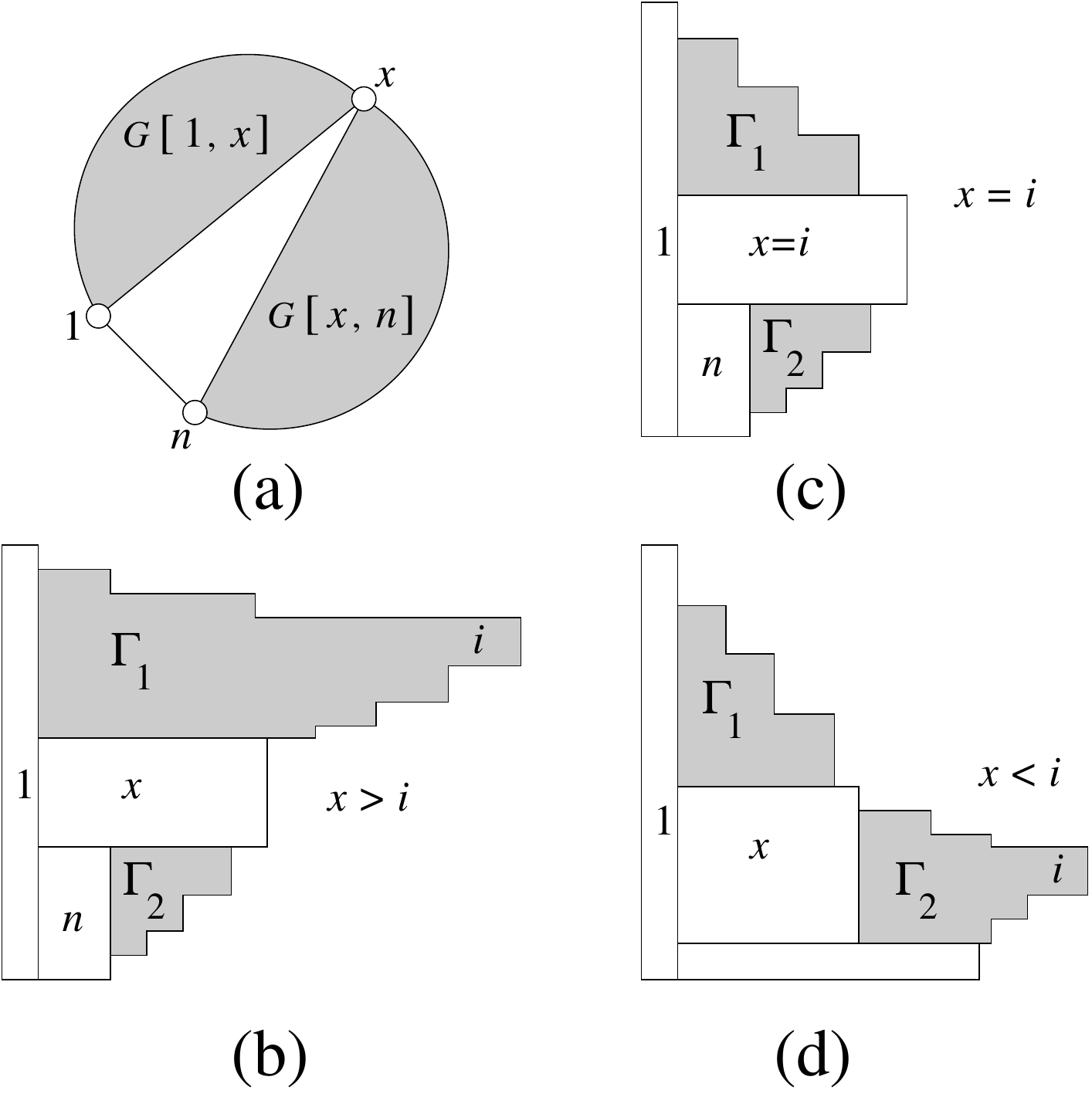}
	\caption{Illustration for the proof of Lemma~\ref{lem:double-stair}.}
\label{fig:double-stair}
\end{minipage}
\hfill
\end{figure}

 Note that in the layout obtained above
%in the proof of Lemma~\ref{lem:stair} for a maximal
 %outerplanar graph, the right sides of the rectangles for the
 %vertices $\{1, \ldots, n\}$
the top right corners of the rectangles for vertices $\{1, \ldots, n\}$
 have increasing $x$-coordinates and decreasing
 $y$-coordinates. Thus we refer to them as \textbf{Staircase} layouts
 and to the algorithm as the \textbf{Staircase
 Algorithm}.

\begin{lemma}
\label{lem:double-stair}
	Let $O$ be a maximal outerplanar graph with a weight function $w$. Let
	$1$, $\ldots$, $n$ be the clockwise order of the vertices
        around the outer-cycle.
Then a proportional rectangle-contact representation  $\Gamma$ of $O$
for $w$ can be computed so that rectangle $R_1$ for $1$ is
leftmost in $\Gamma$, rectangle $R_n$ for $n$ is bottommost in
$\Gamma-R_1$, and the top-right corners of all rectangles
	for vertices $\{1, \ldots, i\}$ and the bottom-right corners of all rectangles for vertices
	$\{i, \ldots, n\}$ are exposed in $\Gamma$.
\end{lemma}
\begin{comment}
\begin{figure}[t]
	%\vspace{-0.8cm}
	\centering
	\includegraphics[width=0.9\textwidth]{double-stair}\\
	(a)\hspace{0.2\textwidth}(b)\hspace{0.22\textwidth}(c)
	\hspace{0.2\textwidth}(d)\hspace{0.05\textwidth}
	\caption{Illustration for the proof of Lemma~\ref{lem:double-stair}.}
\label{fig:double-stair}
\end{figure}

\begin{sketch}
Computing $\Gamma$ is easy when $G$ is a single edge $(1,n)$, so let
 $G$ have at least 3 vertices and $x$ be the unique vertex adjacent to $(1,n)$. Define the two
 graphs $G[1,x]$ and $G[x,n]$ as in the proof of Lemma~\ref{lem:stair}; see Fig.~\ref{fig:double-stair}(a).
If $x>i$, then recursively draw $G[1,x]$ and remove the rectangles for $1$ and $x$ from it;
 call the result $\Gamma_1$. Draw $G[x,n]$ using the \textbf{Staircase Algorithm} and remove
 $x$ and $n$ to find $\Gamma_2$. Now draw three mutually touching rectangles $R_1$, $R_x$
 and $R_n$ for $1$, $x$ and $n$, with the necessary areas and place $\Gamma_1$
 (after possible scaling)  between $R_1$, $R_x$ and $\Gamma_2$ (after $90^{\circ}$ clockwise
 rotation and possible scaling) between $R_x$, $R_n$ to complete the drawing; see
 Fig.~\ref{fig:double-stair}(b). The cases when $x=i$ and $x<i$ follow similar constructions;
 see Fig.~\ref{fig:double-stair}(c)--(d).
\end{sketch}
\end{comment}

\begin{proof}
We again compute $\Gamma$ recursively. Constructing $\Gamma$ is easy when $G$ is a single
 edge $(1,n)$. We thus assume that $G$ has at least 3 vertices. Let $x$ be the (unique)
 third vertex on the inner face that is adjacent to $(1,n)$. Define the two graphs $G[1,x]$
 and $G[x,n]$ as in the proof of Lemma~\ref{lem:stair}; see also Fig.~\ref{fig:double-stair}(a).

%	\begin{figure}[htbp]
%		\centering
%		\includegraphics[width=\textwidth]{double-stair.eps}
%		\caption{Illustration for the proof of Lemma~\ref{lem:double-stair}.}
%	\label{fig:double-stair}
%	\end{figure}

If $x>i$, then recursively draw $G[1,x]$ and remove the rectangles for $1$ and $x$ from it;
 call the result $\Gamma_1$. Draw $G[x,n]$ using the \textbf{Staircase Algorithm}. Now draw
 three mutually touching rectangles $R_1$, $R_x$ and $R_n$ for $1$, $x$ and $n$, respectively
 with necessary areas such that the right side of $R_1$ touches both $R_x$ and $R_n$ and
 the right side of $R_x$ has greater $x$-coordinate than the right side of $R_n$; see
 Fig.~\ref{fig:double-stair}(b). Finally place $\Gamma_1$ (after possible scaling) touching
 the right side of $R_1$ and the top side of $R_x$ such that the right side of the rectangle
 for $(x-1)$ extends past $R_x$. Also place $\Gamma_2$ (after $90^{\circ}$ clockwise rotation
 and possible scaling) touching the bottom side of $R_x$ and the right side of $R_n$ to
 complete the drawing (the width of $R_x$ can be chosen long enough so that $\Gamma_2$
 can be contained between the bottom side of $R_x$ and the right side of $R_n$).

 On the other hand if $x=i$ we follow almost the same procedure as in the previous paragraph.
	However, instead of the drawing of $G[1,x]$ recursively, we compute it by the \textbf{Staircase
	Algorithm} and then delete from it $1$ and $x$ to obtain $\Gamma_1$. We compute $\Gamma_2$
	as in the previous section. We also draw $R_1$, $R_x$ and $R_n$ in the same way. Then we
	place $\Gamma_1$ (after possible scaling) touching the right side of $R_1$ and the top
	side of $R_x$ (again the width of $R_x$ is chosen long enough so that this can be done).
	We also place $\Gamma_2$ as the same manner as before to complete the drawing; see
	Fig.~\ref{fig:double-stair}(c).

	Finally if $x<i$, then we draw $G[1,x]$ by the \textbf{Staircase Algorithm} and delete
	from it $1$ and $x$ to obtain $\Gamma_1$. However, to compute $\Gamma_2$, we recursively
	draw $G[x,n]$ and delete $x$ and $n$ from it. We now draw $R_1$, $R_x$ and $R_n$ as
	before but this time the right side of $R_n$ should extend past $R_x$. We now place
	$\Gamma_1$ (after possible scaling) touching the right side of $R_1$ and the top side
	of $R_x$ (this is again possible for suitable choice of the height of $R_1$). Finally
	we complete the drawing by placing $\Gamma_2$ (after possible scaling) touching the
	right side of $R_x$ and the top side of $R_n$ so that the right side of the rectangle
	for $n-1$ extends past $R_n$; see Fig.~\ref{fig:double-stair}(d).
\end{proof}

Note that in the layout obtained above the top-right corners for vertices $\{1, \ldots, i\}$ and the bottom-right corners for vertices $\{i+1, \ldots, n\}$ form two staircases.
%
% In the layout from the proof of Lemma~\ref{lem:double-stair}, the right sides of the rectangles
% for the vertices $\{1, \ldots, i\}$ have increasing $x$-coordinates while the right side of the
% rectangles for $\{i+1, \ldots, n\}$ have decreasing $x$-coordinates. Also the bottom sides for
% all rectangles have decreasing $y$-coordinates. Thus these layouts forms two separate
% ``staircases'' with a common rectangle for the vertex $i$. Call this layout a
 Thus we refer to this as a \textbf{Double-Staircase} layout,  to the algorithm as the \textbf{Double-Staircase Algorithm}, and to vertex $i$ as the \textit{pivot vertex}.

Let $O$ be a maximal outerplanar graph and let $\Gamma$ be either a \textbf{Staircase} or
 a \textbf{Double-Staircase} layout. Then any triangle $\{p,q,r\}$ in $O$ is represented by three rectangles and the shared boundaries of these rectangles define a \textit{T-shape}. The vertex whose two shared boundaries are  collinear in the T-shape is called the \textit{pole} of the triangle $\{p,q,r\}$.

%We use the two algorithms described in the previous two lemmas to prove %Theorem~\ref{th:nested}.

\smallskip\noindent
\textbf{Proof of Theorem~\ref{th:nested-prop}.}
	Let $\mathcal{T}$ be the construction tree for $G$. We compute a
	representation for $G$ by a top-down traversal of $\mathcal{T}$, constructing the
	representation for each piece as we traverse it. Let $P$ be a piece of $G$ at the
	$l$-th level. If $P$ is the root of $\mathcal{T}$, then we use the
	\textbf{Staircase Algorithm} to find a contact representation of $P$ with rectangles
	in the plane and then we give necessary heights to these rectangles to obtain a proportional
	contact representation of $P$ with boxes.
%	If $P$ is not the root, then
	Otherwise, the vertices of $P$ are adjacent to exactly three $(l-1)$-level vertices
	$A$, $B$, $C$ that form a triangle in the parent piece of $P$.
	Since $A$, $B$, $C$ belong to the parent piece of $P$, their boxes have already been
	drawn when we start to draw $P$. To find a correct
	representation of $G$, we need that the boxes for the vertices in $P$ have correct
	adjacencies with the boxes for
	$A$, $B$, and $C$; hence we assume a fixed structure for such a triangle.
%	In particular,
	We maintain the following invariant:

	\textit{Let $\{p,q,r\}$ be three vertices in a piece $P$ of $G$ forming a triangle.
	 Then in the proportional contact representation of $P$, the boxes for $p$, $q$, $r$
	 are drawn in such a way that (i) the projection of the mutually shared boundaries for
	 these boxes in the $xy$-plane forms a T-shape, (ii) the highest faces (faces
	 with largest $z$-coordinate) of the three rectangles have different $z$ coordinates
	and the highest face of the pole-vertex of the triangle has the smallest
	 $z$-coordinate.}

	Note that by choosing the areas of the rectangles in the \textbf{Staircase} layout,
	we can maintain this invariant for the parent piece by appropriately adjusting the heights of the boxes(e.g., incrementally increasing heights for the vertices in the
	recursive \textbf{Staircase Algorithm}).

	We now describe the construction of a proportional box-contact representation of $P$
	with the correct adjacencies for $A$, $B$ and $C$.
	By the invariant the projection of the shared boundaries for $\{A,B,C\}$ forms a T-shape
	in the $xy$-plane. 
Without loss of generality assume that $A$ is the pole of the triangle and the highest faces
	of $B$, $C$ and $A$ are in this order according to decreasing $z$-coordinates.
%Without loss of generality assume that $B$ has the highest face and $A$ is the pole of this triangle, and so has the lowest face.
%among the three of	$\{A,B,C\}$. Assume without loss of generality that the highest face of $B$ is the highest.
	Also assume that $P$ is a maximal outerplanar graph; we later argue
	that this assumption is not necessary.
	
Let $ab$ be a common neighbor of $A$ and $B$; $bc$ a common neighbor of $B$ and $C$;
	$ca$ a common neighbor of $C$ and $A$. Then the outer cycle of $P$ can be partitioned
	into three paths: $P_a$ is the path from $ca$ to $ab$, $P_b$ is the path from $ab$ to
	$bc$ and $P_c$ is the path from $bc$ to $ca$. All the vertices on the path
	$P_a$ ($P_b$, $P_c$, respectively) are adjacent to $A$ ($B$, $C$, respectively). We
	first assume that there is no chord in $P$ between $ca$ and a vertex on path $P_a$.
	We consider the following two cases.

	\smallskip\noindent
	\textbf{Case 1: No vertex of $P$ is adjacent to all of $\{A,B,C\}$.} We label the vertices
	of $P$ in the clockwise order starting from $ca=1$ and ending at $n$, where $n$ is the
	number of vertices in $P$. Let $i$ and $j$ be the indices of vertices $bc$ and $ab$,
	respectively. Let $x$ be the index of the vertex that is the (unique) third vertex of
	the inner face of $P$ containing the edge $(1,n)$. Define the two graphs $G[1,x]$ and
	$G[x,n]$ as in the proof of Lemma~\ref{lem:stair}. We first find a proportional contact
	representation of $P$ for $w$ restricted to the vertices of $P$
	using rectangles in the plane, then we give necessary heights to this rectangles.
	Draw $G[x,n]$ using the \textbf{Staircase Algorithm} and delete the rectangles for $x$ and
	$n$ to obtain $\Gamma_2$. Draw rectangles $R_x$ and $R_n$ for $x$ and $n$,
	respectively, so that the bottom side of $R_x$ touches the top side of $R_n$, the left
	sides for both the rectangles have the same $x$-coordinate and and the right side of
	$R_n$ extends past $R_x$. Now place $\Gamma_2$ (after possible scaling) touching the
	right side of $R_x$ and the top side of $R_n$ (this is possible since we can make the
	width of $R_n$ sufficiently long); see Fig.~\ref{fig:prop}. Place the rectangle $R_1$
	for $1$ touching the left sides of $R_x$ and $R_n$ such that its bottom side is aligned
	with $R_n$ and its top side is aligned with the top side of the rectangle for $j$.
	To complete the rest of the drawing, we have the following two subcases:
	
	\textbf{Case 1a: $x\le i$.} We draw $G[1,x]$ using the \textbf{Staircase Algorithm} and
	delete from it the rectangles for $1$ and $x$ to obtain $\Gamma_1$. We finally place
	$\Gamma_1$ (after $90^{\circ}$ counterclockwise rotation and possible scaling) touching
	the top side of $R_1$ and left side of $R_x$ (this is possible by choosing a sufficiently
	large height for $R_x$); see Fig.~\ref{fig:prop}(a).

	\textbf{Case 1b: $x>i$.} We draw $G[1,x]$ using the \textbf{Double-Staircase Algorithm}
	where $i$ is the pivot vertex. From this drawing, we delete the rectangles for $1$
and $x$ to obtain $\Gamma_1$. Finally place $\Gamma_1$ (after $90^{\circ}$ counterclockwise
	rotation and possible scaling) touching the top side of $R_1$ and left side of $R_x$ such
	that the topside of the rectangle for $(x-1)$ goes past the top side of $R_x$; see Fig.~\ref{fig:prop}(b).

\begin{figure}[htbp]
\vspace{-1cm}
	\centering
	\includegraphics[width=0.88\textwidth]{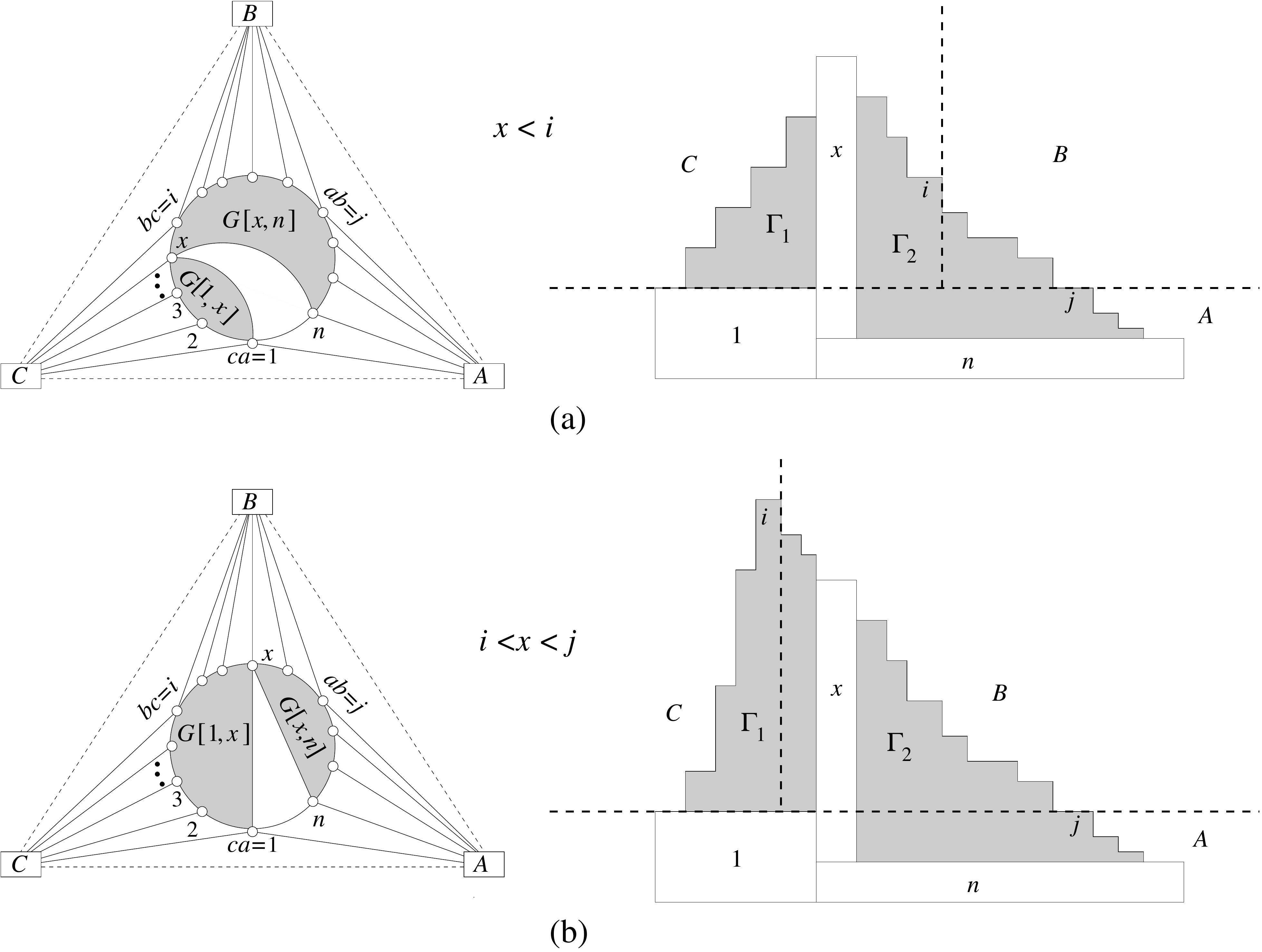}
	\includegraphics[width=0.88\textwidth]{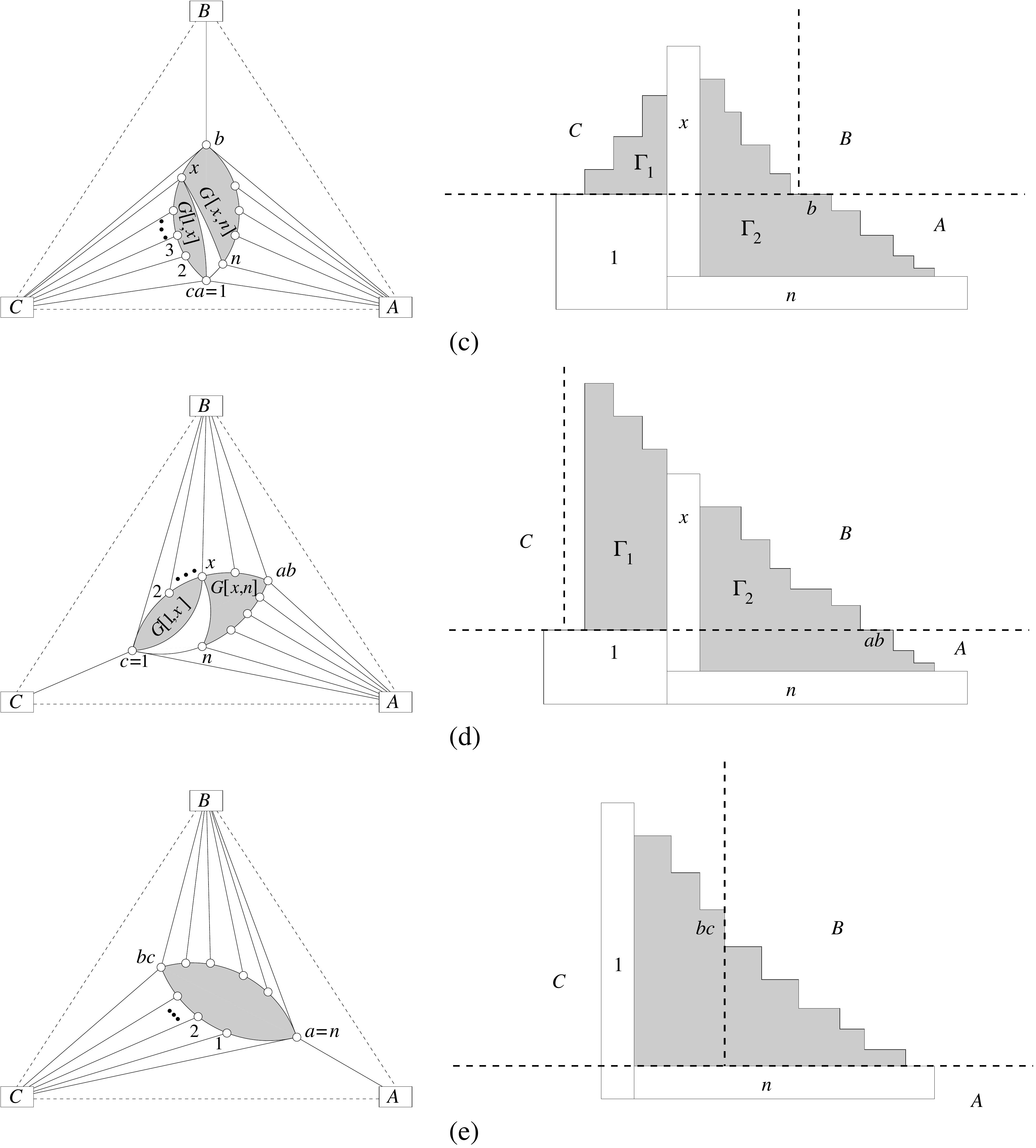}
	\caption{Illustration for Theorem~\ref{th:nested-prop}.
	Construction of the representation for a piece $P$ of $G$,
		%a nested outerplanar graph,
	when (a)--(b) no vertex of $P$ is adjacent to all of $\{A,B,C\}$, and
	(c)--(e) a vertex of $P$ is adjacent to all of $\{A,B,C\}$.}
\label{fig:prop}
\end{figure}

So far we used the function $w$ to assign areas for the rectangles and obtained proportional box-contact representation of $P$ from the rectangles by assigning unit heights. However, by changing the areas for the
	rectangles, we can obtain different heights for the boxes. We will use this property to maintain adjacencies with $\{A,B,C\}$, as well as to maintain the invariant. Specifically, once we get the box
	representation of $P$, we scale it 
%so that its total area becomes very small
%	(with increased heights for the rectangles) 
by increasing the heights for the boxes, so that when we place it at the corner created by
	the T-shape for $\{A,B,C\}$ it will not intersect the representation for
	any of its sibling pieces in $\mathcal{T}$. Consider the point $p$ which is the intersection
	of the lines containing the right side of the rectangle for $i$ and the top side of the
	rectangle for $j$. We place $\Gamma$ such that the point $p$ superimposes
	on the corner for the T-shape in the projection on the $xy$-plane. Since the highest faces
	of $B$, $C$ and $A$ are in this order according to $z$-coordinate,
	the adjacencies of the vertices in $P$ with $\{A,B,C\}$ are correct. By appropriately choosing the areas
	for the rectangles, we ensure that all the boxes for the vertices of $P$
	have their highest faces above that of $B$ and that the invariant is maintained.

%\begin{figure}[htbp]
%	\centering
%	\includegraphics[width=\textwidth]{prop-box_2.eps}
%	\caption{Construction of the proportional contact representation of a piece $P$ of
%		a nested outerplanar graph when a vertex of $P$ is adjacent to all of $\{A,B,C\}$.}
%\label{fig:prop_special}
%\end{figure}

	\noindent
	\textbf{Case 2: A vertex of $P$ is adjacent to all of $\{A,B,C\}$.} In this case at
	least one of $\{A,B,C\}$ has only one neighbor in $P$. Assume first that a vertex $b$
	($=ab=bc$) of $P$ is adjacent to all of $\{A,B,C\}$ and this is the only neighbor of $B$;
	see Fig.~\ref{fig:prop}(c).
 Then we follow the steps for \textit{Case 1a}
	with $b=j$ (and some vertex between $x$ and $b$ as $i$). But when we finally
	place this representation of $P$ on the corner for the T-shape of $\{A,B,C\}$ we find
	the point $p$ to superimpose on this corner as follows. The point $p$ is on the line
	containing the top side of the rectangle for $b$ and has $x$-coordinate between the
	right sides of the rectangles for $b$ and $(b-1)$.%, respectively.
	
	If a vertex $c$ ($=bc=ca$) is adjacent to all of $\{A,B,C\}$ and is
	the only neighbor of $C$ in $P$, then we follow the steps of \textit{Case 1b} with $i=2$
	and $j=ab$; see Fig.~\ref{fig:prop}(d).
 We find the point $p$ to superimpose on
	the corner for the T-shape of $\{A,B,C\}$ as follows. The point $p$ is on the line
	containing the top side of the rectangle for $j=ab$ and has $x$-coordinate between
	the left sides of the rectangles for $1$ and $2$.%, respectively.

	If a vertex $a$ ($=ab=ca$) is adjacent to all of $\{A,B,C\}$ and is
	the only neighbor of $A$ in $P$, then we number the vertices of $P$ in the clockwise
	order starting from the clockwise neighbor of $a$ and ending at $a=n$;
	see Fig.~\ref{fig:prop}(e).
	 We use the
	\textbf{Staircase Algorithm} to find a representation of $P$ with
	rectangles and give necessary heights to obtain a representation
	with boxes. On the corner for the T-shape of $\{A,B,C\}$, we superimpose the
	intersection point for the lines containing the top side of the rectangle of $n$ and the
	right side of the rectangle for $bc$.

	Finally, we consider the case when there is a chord between $ca$ and another vertex on the
	path $P_a$. Take the innermost such chord and let its other end-vertex be $t$. Then consider the two subgraphs $P_1$
	and $P_2$ induced by all the vertices outside the chord and inside the chord (along with
	the two vertices $ca$ and $t$). $P_1$ does not contain any chord from $ca$;
	thus we use the algorithm above %described so far 
to obtain a
	representation of $P_1$; denote this by $\Gamma'$. In this
	representation $ca$ and $t$ will play the roles of $1$ and $n$, respectively.
	Each vertex of $P_2$ is adjacent to $A$ and we find a proportional contact
	representation of $P_2$ and attach it with $\Gamma'$ as follows. We use the \textbf{Staircase
	Algorithm} to find a proportional contact representation of $P_2$ with rectangles in
	the plane and delete the rectangles for $ca$ and $t$ from it to obtain $\Gamma''$.
	In $\Gamma'$, we change the height of the rectangle $R_1$ for $ca=1$ to increase
	its area so that its bottom side extends past the bottom side of the rectangle $R_n$
	for $t=n$. Then we place $\Gamma''$ (after reflecting with respect to the $x$-axis
	and possible scaling) touching the right side of $R_1$ and the bottom side of $R_n$.
	Since the \textbf{Staircase Algorithm} can accommodate any given area for the layout,
	we can change the heights of the boxes for the vertices in $P_2$ to maintain the invariant.

	Thus with the top-down traversal of $\mathcal{T}$, we obtain a proportional contact
	representation for $O$. We assumed that each piece of $O$ is maximal outerplanar.
	However in the contact representation, for each edge $(u,v)$, either
	a face of the box $R_u$ for $u$ is adjacent to the box $R_v$ for $v$ and no other box;
	or a face of $R_v$ is adjacent to $R_u$ and no other box. In both cases the adjacency
	between these two coxes can be removed without affecting any other adjacency. Thus this
	algorithm holds for any nested outerplanar graph $O$. \qed

\section{Conclusions and Future Work}

%when proper contacts are required. 
We proved that nested maximal outerplanar graphs have cube-contact
 representations and nested outerplanar graphs have proportional box-contact representations.
 These classes of graphs are special cases of $k$-outerplanar graphs, and the set of
 $k$-outerplanar graphs for all $k>0$ is equivalent to the class of all planar graphs. 
Even though our approach might generalize to large classes, 
cube-contact representations and proportional box-contact representations are still open for general planar graphs.
%Thus,
% this approach might eventually help showing that all planar graphs have proper cube-contact
% or proportional box-contact representations.

\medskip
\noindent{\bf Acknowledgments:} We thank 
 Therese Biedl, Steve Chaplick, Stefan Felsner, and Torsten Ueckerdt for discussions about this problem.

{
\begin{small}

%\bibliography{contact.bib}
%\bibliography{stephen_more}
\end{small}
}

%\newpage
%\input{appendix.tex}

\end{document}